\documentclass[a4paper,11pt,amsthm,abbrv]{llncs}

\usepackage{amsmath}
\usepackage{amsfonts}
\usepackage{graphicx} 
\usepackage{subfigure} 
\usepackage{llncsdoc}
   
\newcommand{\eps}{\varepsilon}

\newcommand{\complexity}[1]{\ensuremath{\mathcal{#1}}}
\newcommand{\height}[1]{\ensuremath{\MakeLowercase{#1}}}
\newcommand{\width}[1]{\ensuremath{\MakeLowercase{#1}}}

\newcommand{\citesuperscript}[1]{\textsuperscript{\cite{#1}}}  	

\newcommand\opt{\mathrm{OPT}}
\newcommand{\Opt}[1]{\ensuremath{\mathrm{OPT}(#1)}}
\newcommand{\Oh}[1]{\ensuremath{\mathcal{O}(#1)}} 
\newcommand{\Vol}[1]{\ensuremath{\mathcal{A}(#1)}}
\newcommand{\PTAS}{\ensuremath{\mathcal{PT\hspace{-.25em}AS}}}
\newcommand{\FPTAS}{\ensuremath{\mathcal{FPT\hspace{-.25em}AS}}}

\spnewtheorem{observation}{Observation}{\bf}{\it}

\usepackage[pdftex]{hyperref}	
\hypersetup{%
pdftitle={An Absolute 2-Approximation Algorithm for Two-Dimensional Bin Packing}, 
pdfauthor={Rolf Harren}, 
pdfkeywords={}, 
colorlinks={true},
linkcolor={red},  
citecolor={blue},
filecolor={magenta},
urlcolor={cyan}
}

\begin{document}

\title{An Absolute 2-Approximation Algorithm for Two-Dimensional Bin Packing}

\author{Rolf Harren and Rob van Stee\thanks{Research supported by German Research Foundation (DFG)}}
\institute{Max-Planck-Institut f\"ur Informatik (MPII), \\ Campus E1 4, D-66123 Saarbr\"ucken, Germany. \\ \email{\{rharren,vanstee\}@mpi-inf.mpg.de}}
\maketitle

\begin{abstract}
We consider the problem of packing rectangles into bins that are unit squares, where the goal is to minimize the number of bins used. All rectangles have to be packed non-overlapping and orthogonal, i.e., axis-parallel. We present an algorithm for this problem with an absolute worst-case ratio of 2, which is optimal provided $\mathcal{P} \not= \mathcal{NP}$.

\end{abstract}

\newcommand{\alg}{\mbox{\textsc{alg}}}
\newcommand{\A}{{\cal A}}

\section{Introduction}

In the two-dimensional bin packing problem, a list $I= \{r_1, \ldots, r_n\}$ of rectangles of width $w_i \leq 1$ and height $h_i \leq 1$ is given. An unlimited supply of unit-sized bins is available to pack all items from $I$ such that no two items overlap and all items are packed axis-parallel into the bins. The goal is to minimize the number of bins used. The problem has many applications, for instance in stock-cutting or scheduling on partitionable resources. In many applications, rotations are not allowed because of the pattern of the cloth or the grain of the wood. This is the case we consider in this paper.

Most of the previous work on rectangle packing has focused on the \emph{asymptotic} approximation ratio, i.e., the long-term behavior of the algorithm. 
The asymptotic approximation ratio is defined as follows. 
Let $\alg(I)$ be the number of bins used by algorithm
$\alg$ to pack input $I$. Denote the optimal number of bins by $\Opt{I}$.
The asymptotic approximation ratio
of a two-dimensional bin packing algorithm $\alg$ is defined to be
$$\limsup_{n\to\infty} \sup_I\left\{\left. \frac{\alg(I)}{\opt(I)}\ \right|\ 
\opt(I)=n\right\}.$$
Caprara\citesuperscript{Caprara:2002a} was the first to present an algorithm with an asymptotic approximation ratio less than 2 for two-dimensional bin packing. Indeed, he considered 2-stage packing, in which the items must first be packed into shelves that are then packed into bins, and showed that the asymptotic worst case ratio between two-dimensional bin packing and 2-stage packing is $T_{\infty} = 1.691\ldots$. Therefore the asymptotic \FPTAS\ for 2-stage packing from Caprara, Lodi \& Monaci\citesuperscript{CapraraLodiMonaci:2005a} achieves an asymptotic approximation guarantee arbitrary close to $T_{\infty}$. 

Recently, Bansal, Caprara \& Sviridenko\citesuperscript{BansalCapraraSviridenko:2006a} presented a general framework to improve subset oblivious algorithms and obtained asymptotic approximation guarantees arbitrarily close to $1.525\ldots$ for packing with rotations of 90 degrees or without rotations. These are the currently best-known asymptotic approximation ratios for these problems. For packing squares into square bins, Bansal, Correa, Kenyon \& Sviridenko\citesuperscript{BansalCorreaKenyonSviridenko:2006a} gave an asymptotic \PTAS. On the other hand, the same paper showed the $\mathcal{APX}$-hardness of two-dimensional bin packing without rotations, thus no asymptotic \PTAS\ exists unless $\mathcal{P} = \mathcal{NP}$. 
Chleb\'{\i}k \& Chleb\'{\i}kov\'a\citesuperscript{ChlebikChlebikova:2005a} were the first to give explicit lower bounds of $1 + 1/3792$ and $1+ 1/2196$ on the asymptotic approximability of rectangle packing with and without rotations, respectively.

It should be noted that for some of the positive results mentioned above, 
the approximation ratio only gets close to the stated value for very large
inputs. In particular, the $1.525$-approximation from Bansal et al.\citesuperscript{BansalCapraraSviridenko:2006a} has an additive constant which is not
made explicit in the paper but which the authors believe is extremely 
large\citesuperscript{Bansal:2008}.
Thus, for any reasonable input, the actual (absolute) approximation ratio of
their algorithm is much larger than $1.525$, and it therefore makes sense to
consider alternative algorithms and in particular, an alternative performance
measure.

In the current paper, we consider the absolute approximation ratio. 
This is defined simply as $\sup_I \alg(I)/\opt(I)$, where the supremum is taken
over all possible inputs. Proving a bound on the absolute approximation gives
us a performance guarantee for all inputs, not just for (very) large ones.
Zhang\citesuperscript{Zhang:2005a} presented an approximation algorithm with an absolute approximation ratio of $3$ for the problem without rotations. For the special case of packing squares, van Stee\citesuperscript{Stee:2004a} showed that an absolute 2-approximation is possible. Additionally, Harren \& van Stee\citesuperscript{HarrenvanStee:2009} presented an
absolute 2-approximation for packing rectangles in the case that 90 degree rotations
are allowed. They also showed that the algorithm Hybrid First Fit has an absolute
approximation ratio of 3 for packing without rotations, as conjectured by Zhang\citesuperscript{Zhang:2005a}.

A related two-dimensional packing problem is the strip packing problem, where the items have to be packed into a strip of unit basis and unlimited height such that the height is minimized. Steinberg\citesuperscript{Steinberg:1997a} and Schiermeyer\citesuperscript{Schiermeyer:1994a} presented absolute 2-approximation algorithms for strip packing without rotations. 
We use Steinberg's algorithm in particular as a subroutine in our algorithm.

\paragraph{Our contribution.} We present an approximation algorithm for two-dimensional bin packing with an absolute approximation ratio of 2. As Leung et al.\citesuperscript{LeungTamWongYoungChin:1990a} showed that it is strongly \complexity{NP}-complete to decide wether a set of \emph{squares} can be packed into a given square, this is best possible unless $\mathcal{P} = \mathcal{NP}$.

The two-dimensional bin packing problem can also be seen as a scheduling problem. Here the tasks $r_i = (w_i,h_i)$ have a running time $w_i$ and need $h_i$ consecutive machines to be proceeded. The restriction to a polynomial number of machines can be solved with methods from Jansen \& Th\"ole\citesuperscript{JansenThole:2008}, but the unrestricted version that we consider in this paper appears to be considerably more difficult.

\section{Important tools and definitions}\label{sec:tools}
In this section we give the necessary definitions and introduce results that are important for our work. 

Let $I= \{r_1, \ldots, r_n\}$ be the set of given rectangles, where $r_i = (w_i, h_i)$. Let $W_{\delta} = \{r_i \mid w_i > 1-\delta\}$ be the set of so-called \emph{$\delta$-wide} items and let $H_{\delta} = \{r_i \mid h_i > 1-\delta\}$ be the set of \emph{$\delta$-high} items. To simplify the presentation, we denote the $1/2$-wide items as \emph{wide} items and the $1/2$-high items as \emph{high} items. Let $W$ and $H$ be the sets of wide and high items, respectively. The set of \emph{small} items, i.e., items $r_i$ with $w_i \leq 1/2$ and $h_i \leq 1/2$, is denoted by $S$. Finally, we call items that are wide and high at the same time \emph{big}.

For a set $T$ of items, let $\Vol{T} = \sum_{i\in T} w_i h_i$ be the total area and let $\height{h}(T) = \sum_{r_i \in T} h_i$ and $\width{w}(T) = \sum_{r_i \in T} w_i$ be the total height and total width, respectively. Finally, let $w_{\max}(T) = \max_{r_i \in T} w_i$ and $h_{\max}(T) = \max_{r_i \in T} h_i$.

Steinberg\citesuperscript{Steinberg:1997a} proved the following theorem for his algorithm that we use as a subroutine. 
\begin{theorem}[Steinberg's algorithm] \label{thm:Steinberg} 
If the following inequalities hold,
\begin{displaymath}
w_{\max}(T) \leq a, \quad h_{\max}(T)Ê\leq b, \textrm{ and} \quad 2 \Vol{T} \leq ab - (2w_{\max}(T) - a)_{+} (2h_{\max}(T) - b)_{+}
\end{displaymath}
where $x_{+} =\max(x, 0)$, 
then it is possible to pack all items from $T$  into $R = (a,b)$ in time $\Oh{(n \log^2 n) / \log \log n}$. 
\end{theorem}

Jansen \& Zhang\citesuperscript{JansenZhang:2007a} showed the following corollary from this theorem.
\begin{corollary}\label{cor:big}
If the total area of a set $T$ of items is at most $1/2$ and there are no wide items (except a possible big item) then the items in $T$ can be packed into a bin.
\end{corollary}
Obviously, the corollary also holds for the case that there are no high items (except a possible big item). This corollary is an improvement upon Theorem~\ref{thm:Steinberg} if there is a big item in $T$ as in this case Theorem~\ref{thm:Steinberg} would give a worse area bound.

Bansal, Caprara \& Sviridenko\citesuperscript{BansalCapraraSviridenko:2008a} considered the two-dimensional knapsack problem in which each item $r_i \in I$ has an associated profit $p_i$ and the goal is to maximize the total profit that is packed into a unit-sized bin. Using a very technical \emph{Structural Lemma} they derived an algorithm that we call BCS algorithm in this paper. They showed the following theorem.

\begin{theorem}[Bansal, Caprara \& Sviridenko]\label{thm:generalPTAS}
For any fixed $r \geq 1$ and $\delta > 0$, the BCS algorithm returns a packing of value at least $(1-\eps)\opt_{2KP}(I) - \eps$ for instances $I$ for which $p_i/\Vol{r_i} \in [1,r]$ for $r_i \in I$. The running time of the BCS algorithm is polynomial in the number of items.
\end{theorem}
Here $\opt_{2KP}(I)$ denotes the maximal profit that can be packed in a bin of unit size.
In the case that $p_i = \Vol{r_i}$ we want to maximize the total packed area. Let $\opt_{(a,b)}(T)$ denote the maximum area of items from $T$ that can be packed into the rectangle $(a,b)$, where individual items in $T$ do not necessarily fit in $(a,b)$. By appropriately scaling the bin, the items and the accuracy we get the following corollary.

\begin{corollary}[Bansal, Caprara \& Sviridenko]\label{thm:PTAS}
For any fixed $\eps > 0$, the BCS algorithm returns a packing of $I' \subseteq I$ in a bin of width $a\leq 1$ and height $b\leq 1$ such that $\Vol{I'}Ê\geq \opt_{(a,b)}(I) - \eps$.
\end{corollary}

\section{Our algorithm: Overview}

As the asymptotic approximation ratio of the algorithm from Bansal, Caprara \& Sviridenko\citesuperscript{BansalCapraraSviridenko:2006a} is arbitrarily close to $1.525\ldots$, there exists a constant $k$ such that for any instance $I$ with optimal value larger than $k$, their algorithm gives a solution of value at most $2\,\opt(I)$. As we already mentioned in the introduction, this constant $k$ is not explicitly known.
We address the problem of approximating the problem within an absolute factor of 2, provided that the optimal value of the given instance is less than $k$. 
Combined with the algorithm from Bansal et al.\ we get an overall algorithm with an absolute approximation ratio of 2. 

Our approach for packing instances $I$ with $\Opt{I} \leq k$ consists of two parts. First, we give an algorithm that is able to pack instances $I$ with $\Opt{I} = 1$ in two bins and second, we show how to approximate instances with $1 < \Opt{I} < k$ within a factor of 2. This at first glance surprising distinction is due to the inherent difficulty of packing wide and high items together into a single bin. In the case $\Opt{I}Ê= 1$ we can not ensure a separation of the wide and high items into feasible sets whereas for $\Opt{I} > 1$ this is possible in most cases. 

\paragraph{Organisation.} In Section~\ref{sec:Opt1} we present the algorithm for packing instances that fit into one bin. In a first step we show how to pack these instances in two bins if the overall height of the $\delta$-wide items is suitably bounded by a function in $\delta$ for some $\delta \in (\eps, 1/2]$. If this is not the case we are able to derive substantial area guarantees for the sets of wide and high items. These area guarantees are utilized  in the second step to pack all other instances $I$ with $\Opt{I} = 1$.

The algorithm for instances $I$ with optimal value within $2$ and $k$ is presented in Section~\ref{sec:OptConst}. The basic algorithm consists of a large enumeration of items with area at least $\eps$ (for some constant $\eps > 0$) and a separation of the wide and high items. On this basis we distinguish a total of four main cases that require different methods to be packed.

Finally, in Section~\ref{sec:final} we put the different parts of the algorithm together.

\section{Packing instances that fit into one bin}
\label{sec:Opt1}
Throughout this section we assume that the given instance $I$ can be packed into a single bin, i.e., $\Opt{I} = 1$. At first glance it seems surprising that packing such an instance into two bins is difficult. Still, we need to carefully analyse different cases to be able to give a polynomial-time algorithm that solves this problem.

Let $0 < \eps < 1/200$ be a fixed constant. In a first step we consider instances that satisfy an upper bound for the total height of the $\delta$-wide items for some $\delta \in (\eps, 1/2]$. Using Steinberg's algorithm and the BCS algorithm we are able to pack these instances into two bins. In a second step we benefit from the upper bounds, as the instances that remain violate these bounds for \emph{all} $\delta \in (\eps, 1/2]$ and thus we have wide items of large total area. We use this area guarantee to give different methods to pack according to the total height of the wide items.

\subsection{Small total height of the $\delta$-wide items}

We want to derive a packing of $I$ into two bins in the case that the total height of the $\delta$-wide items is small relative to $\delta$, i.e., $\height{h}(W_{\delta}) \leq (\delta  - \eps)/(1+2\delta)$ for some $\delta \in (\eps,1/2]$. 

For the ease of presentation let $\gamma = (\delta  - \eps)/(1+2\delta) \leq 1/4$ and assume that $\height{h}(W_{\delta}) \leq \gamma$. In a first step, we show that an almost optimal solution with a special structure exists. This special structure consists of a part of width $\width{w}(H_{\gamma})$ for the $\gamma$-high items and a part of width $1-\width{w}(H_{\gamma})$ for the other items. The following lemma shows that almost all other items can be packed.

\begin{lemma} \label{lem:OPTsmallHeight}
We have $\opt_{(1-\width{w}(H_{\gamma}), 1)}(I \setminus H_{\gamma}) \geq \Vol{I \setminus H_{\gamma}} - 2\gamma$.
\end{lemma}

\begin{proof}
Consider an optimal packing of $I$. 
Remove all items that are completely contained in the top or bottom $\gamma$-margin. After this step there is no item directly above or below any item of $H_{\gamma} = \{r_i \mid h_i > 1-\gamma\}$. Thus we can cut the remaining packing at the left and right side of any items from $H_{\gamma}$. These cuts partition the packing into parts which can be swapped without losing any further item. Move all items of $H_{\gamma}$ to the left of the bin and move all other parts of the packing to the right. The total area of the removed items is at most $2\gamma$ and thus a total area of at least $\Vol{I\setminus H_{\gamma}} - 2\gamma$ fits into a bin of size $(1-\width{w}(H_{\gamma}),1)$. \qed
\end{proof}

In the second step, we actually derive a feasible packing that is based on the structure described above. First pack $H_{\gamma}$ into a stack of width $\width{w}(H_{\gamma})$ at the left side of the first bin. Note that $\width{w}(H_{\gamma}) \leq 1$. This leaves an empty space of width $1-\width{w}(H_{\gamma})$ and height $1$ at the right. We therefore apply the BCS algorithm on $I \setminus H_{\gamma}$ and a bin of size $(1-\width{w}(H_{\gamma}), 1)$ using an accuracy of $\eps$. Lemma~\ref{lem:OPTsmallHeight} and Corollary~\ref{thm:PTAS} yield that at least a total area of $\Vol{I \setminus H_{\gamma}} - 2\gamma - \eps$ is packed by the algorithm.

Let $T$ be the set of remaining items. We have $\Vol{T} \leq 2\gamma + \eps$. Pack the remaining $\delta$-wide items, i.e., the items of $T \cap W_{\delta}$, in a stack at the bottom of the second bin. Let $\height{h}' = \height{h}(T \cap W_{\delta})$ be the total height of these items. We have  $\height{h}' \leq \height{h}(W_\delta) \leq \gamma$. Let $T' = T \setminus W_{\delta}$ be the set of remaining items after this step. Then for $r_i \in T'$ with $r_i = (w_i, h_i)$ we have $w_i \leq 1-\delta$ and $h_i \leq 1-\gamma \leq 1-\height{h}'$. The total area of the remaining items is 
\begin{align*}
\Vol{T'} \leq \Vol{T} - (1-\delta)\height{h}' \leq 2\gamma + \eps - \height{h}' + \delta\height{h}'.
\end{align*}
We pack these items with Steinberg's algorithm into the free rectangle of size $(a,b)$ with $a=1$ and $b = 1-\height{h}'$ at the top of the second bin. In the following we show that this is possible by verifying the conditions of Theorem~\ref{thm:Steinberg}. We already stated that $w_{\max}(T') \leq 1 - \delta \leq 1$ and $h_{\max}(T') \leq 1 - \gamma \leq 1 - \height{h}'$. To show $2\Vol{T'} \leq ab - (2w_{\max} - a)_{+} (2h_{\max} - b)_{+}$ we use 
\begin{align*}
2\Vol{T'} \leq 2(2\gamma + \eps - \height{h}'+ \delta\height{h}')
\end{align*}
and (since $\delta \leq 1/2$ and therefore $\gamma \leq 1/4$)
\begin{align*}
ab - (2w_{\max} - a)_{+} (2h_{\max} - b)_{+} & \geq 1-\height{h}' - (2(1-\delta)-1)_+(2(1-\gamma) - (1-\height{h}'))_+  \\
&= 1-\height{h}' - (1-2\delta)_+(1-2\gamma+\height{h}')_+ \\ 
&= 1-\height{h}' - (1-2\delta - 2\gamma + 4\delta\gamma + \height{h}' - 2 \delta\height{h}') \\
&= 2(\delta + \gamma - 2\delta\gamma - \height{h}' + \delta\height{h}').
\end{align*}
Starting with the definition of $\gamma$ we get
\begin{align*}
&& \gamma & = \frac{\delta -\eps}{1+2\delta} \\
\Leftrightarrow && \gamma(1+2\delta) & = \delta - \eps \\
\Leftrightarrow && \gamma + \eps + 2\delta\gamma & = \delta \\
\Leftrightarrow && 2\gamma + \eps - \height{h}' + \delta\height{h}' & = \delta + \gamma - 2\delta\gamma - \height{h}' + \delta\height{h}' \\
\Rightarrow && 2\Vol{T'} & \leq ab - (2w_{\max} - a)_{+} (2h_{\max} - b)_{+}.
\end{align*}

So far we assumed the knowledge of $\delta \in (\eps, 1/2]$ for which $h(W_{\delta}) \leq (\delta - \eps)/(1+2\delta)$. It is easy to see that this value can be computed by calculating $h(W_{\delta})$ for $\delta = 1-w_i$ for all $r_i = (w_i, h_i)$ with $w_i > 1/2$. As $h(W_{\delta})$ changes only for these values of $\delta$, we will necessarily find a suitable $\delta$ if one exists. We therefore have the following lemma. 
\begin{lemma}\label{lem:smallheight}
For any fixed $\eps > 0$, there exists a polynomial-time algorithm that, given an instance $I$ with $\Opt{I}Ê= 1$ and $\height{h}(W_{\delta}) \leq (\delta  - \eps)/(1+2\delta)$ for some $\delta \in (\eps, 1/2]$, returns a packing of $I$ into two bins.
\end{lemma}

\subsection{Using an area bound for the wide and high items}

In this section we describe how to derive a guarantee on the total area of the wide and high items for the instances that cannot be solved by Lemma~\ref{lem:smallheight}.
Consider a unit-sized bin with the lower left corner at the origin of a cartesian coordinate system and consider the stack of wide items ordered by non-increasing width and aligned with the lower right corner of the bin. If there exists a $\delta \in (\eps, 1/2]$ such that $\height{h}(W_{\delta}) \leq (\delta-\eps)/(1+2\delta)$ then we use the algorithm of Lemma~\ref{lem:smallheight} to pack the instance into two bins--see Figure~\ref{fig:LowerBoundWcase1}. Otherwise the stack of wide items exceeds the function $f(x) = (x-\eps)/(1+2x)$ for $x \in (\eps, 1/2]$--see Figure~\ref{fig:LowerBoundWcase2}. Then for $\eps \leq 1/200$ and $\xi = 0.075$ the total area of the wide items is 
\begin{align} 
\Vol{W} & \geq \int_{\eps}^{1/2} \frac{x-\eps}{1+2x} dx + \frac{\height{h}(W)}{2} \nonumber \\
& \geq \bigg[\frac{1}{4}\Big(2x-(2\eps+1)\log(2x+1)+1\Big)\bigg]_{\eps}^{1/2} + \frac{\height{h}(W)}{2} \nonumber \\
& \geq \xi + \frac{\height{h}(W)}{2}. \nonumber 
\end{align}
Furthermore, for $\delta = 1/2$ we have
\begin{align}
\height{h}(W) > \frac{1/2 - \eps}{1+2 \cdot 1/2} = \frac{1}{4} - \frac{\eps}{2}. \label{eqn:LowerBoundTotalH}
\end{align}

\begin{figure} [!tbp]
    \centering
     \subfigure[If there exists a $\eps < \delta \leq 1/2$ such that $H(W_{\delta}) \leq (\delta-\eps)/(1+2\delta)$ then use Lemma~\ref{lem:smallheight} to pack the instance into two bins. \label{fig:LowerBoundWcase1}]{
     \includegraphics[width=.45\textwidth]{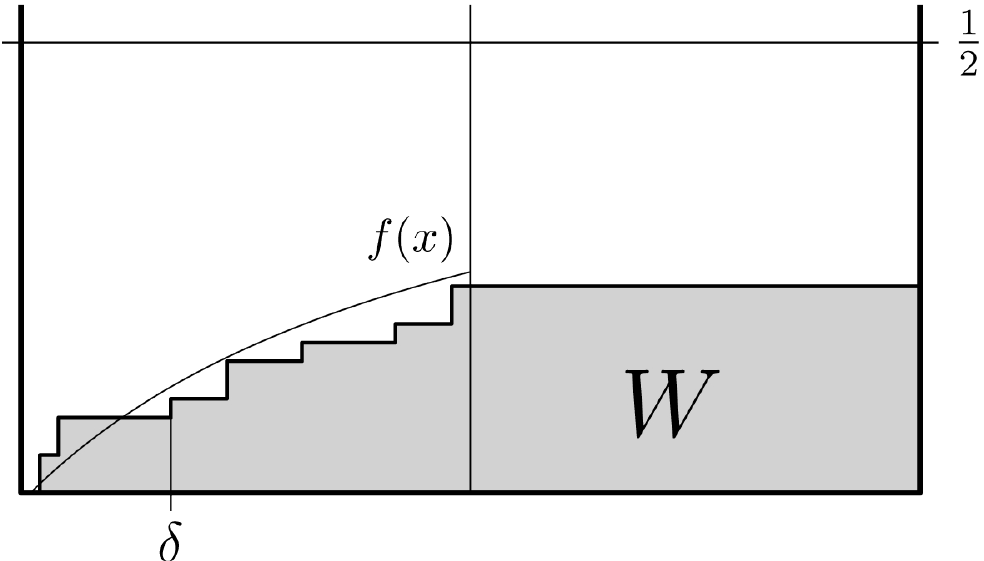}
     }
     \hspace{.05\textwidth}
     \subfigure[Otherwise we get $\Vol{W} \geq \xi + \height{h}(W)/2$ for $\xi = 0.075$. \label{fig:LowerBoundWcase2}]{
      \includegraphics[width=.45\textwidth]{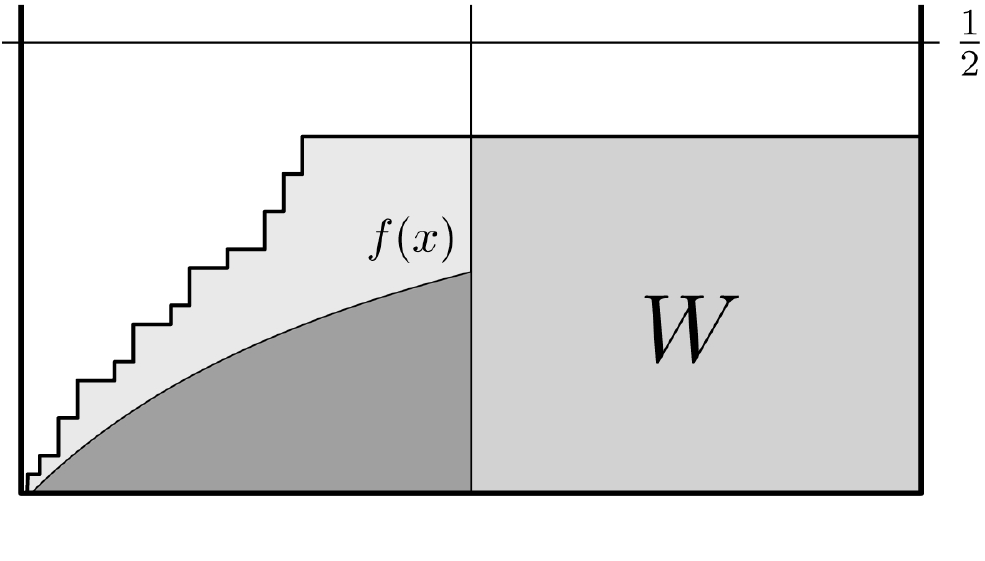}
      }
  \caption{Deriving the lower bound for the area of the wide items. \label{fig:LowerBoundW}}
\end{figure}

Obviously, we can apply Lemma~\ref{lem:smallheight} to the high items instead of the wide items as well. We get the following lemma.
\begin{lemma}
For any input which cannot be packed by the methods described above, we have
\begin{align*}
\Vol{W \cup H} \geq 2\xi + \frac{\width{w}(H) + \height{h}(W)}{2}.
\end{align*}\label{lem:areaofWH}
\end{lemma}
It is crucial for our work that we get this additional area guarantee of $2\xi = 0.15$ on top of the trivial guarantee of $\width{w}(H)/2 + \height{h}(W)/2$ in the previous lemma. 

Let us recapitulate the general idea of the algorithm. If $\height{h}(W_{\delta}) \leq (\delta  - \eps)/(1+2\delta)$ or $\width{w}(H_{\delta}) \leq (\delta  - \eps)/(1+2\delta)$ for some $\delta \in (\eps, 1/2]$ we pack $I$ in two bins using the methods from Lemma~\ref{lem:smallheight}. Otherwise we rotate the instance such that we have $\height{h}(W) \geq \width{w}(H)$ and use the area guarantee of Lemma~\ref{lem:areaofWH} to apply different methods for $\width{w}(H) > 1/2$ and $\width{w}(H) \leq 1/2$. In all cases we are able to pack into at most two bins. Before we can give the main lemmas to solve both cases above we need some preparations.

In the following we show the existence of a packing of the wide items and of high items with at least half of their total width of the high items with a nice structure, i.e., such that the wide and the high items are packed in stacks in different corners of the bin. We later use this observation to approximate such a packing.
\begin{lemma}\label{lem:wideandhigh}
For sets $W$ and $H$ of wide and high items with $\Opt{W \cup H} = 1$ there exists a packing of $W \cup H^*$ with $H^* \subseteq H$ and $\width{w}(H^*) \geq \width{w}(H)/2$ such that the wide items are stacked in the bottom right and the items from $H^*$ are stacked in the top left corner of the bin.
\end{lemma}

\begin{proof}
Consider a packing of wide items $W$ and high items $H$ into a bin. Associate each high item with the nearest side of the bin (an item that has the same distance to both sides of the bin can be associated with an arbitrary side). Assume w.l.o.g. that the total width of the items associated with the left is at least as large as the total width of the items associated with the right. Remove the items that are associated with the right and denote the other high items by $H^*$. Push the items of $H^*$ together into a stack that is aligned with the left side of the bin by moving them purely horizontal and move the wide items such that they are aligned with the right side of the bin and form stacks at the bottom and top side of the bin. Order the stacks of the wide items by non-increasing order of width and the stack of the high items by non-increasing order of height. See left part of Figure~\ref{fig:Insertion_Process}.

Now apply the following process--see also Figure~\ref{fig:Insertion_Process}. Take the shortest item with respect to the width from the top stack of the wide items and insert it at the correct position into the bottom stack, i.e., such that the stack remains in the order of non-increasing width. Move the high items upwards if this insertion causes an overlap. Obviously this process moves all wide items to the bottom and retains a feasible packing.
In the end, all wide items form a stack in the bottom right corner of the bin. Move the high items upwards such that they form a stack in the top left corner of the bin.\qed
\end{proof}
\begin{figure}  [!tbp]
  \centering
  \includegraphics[width=.65\textwidth]{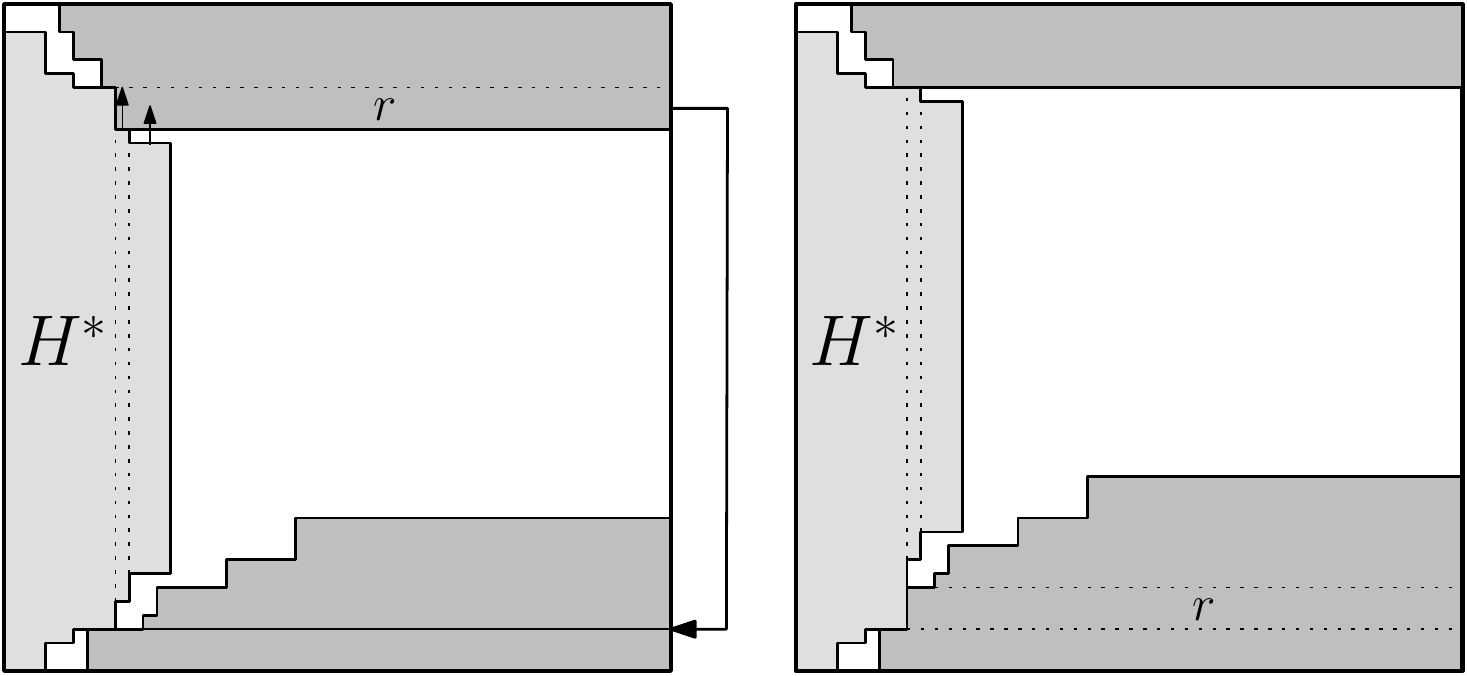}
  \caption{Inserting item $r$ into the bottom stack and moving the high items upwards\label{fig:Insertion_Process}}
\end{figure}
We use this lemma to find a packing of the wide items together with high items of almost half their total width.

\begin{lemma}\label{lem:PackWH}
For any fixed $\eps > 0$, there exists a polynomial-time algorithm that, given sets $W$ and $H$ of wide and high items with $\Opt{W \cup H}Ê= 1$, returns a packing of $W \cup H'$ into a bin with $H' \subseteq H$ and $\width{w}(H') > \width{w}(H)/2 -\eps$.
\end{lemma}

\begin{proof}
By Lemma~\ref{lem:wideandhigh}, a packing of $W\cup H^*$ exists. Let $H^*_{\geq \eps}Ê= \{ r_i \in H^* \mid w_i \geq \eps\}$ and $H^*_{< \eps}Ê= \{ r_i \in H^* \mid w_i < \eps\}$. We approximate $H^* = H^*_{\geq \eps} \cup H^*_{< \eps}$ as follows.

First, pack the items of $W$ in a stack by non-increasing order of width and align this stack with the bottom right corner of the bin. Second, guess the set $H^*_{\geq \eps}$. By guessing we mean that we enumerate all subsets of $\{r_i \in H \mid w_i \geq \eps\}$ (which is possible as $|\{r_i \in H \mid w_i \geq \eps\}| \leq 1/\eps$) and apply the remainder of this algorithm on all these sets. As we eventually consider $H^*_{\geq \eps}$, we assume that we can guess this set. Pack $H^*_{\geq \eps}$ into a stack by non-increasing order of height and align this stack with the top left corner of the bin. Third, we approximate $H^*_{< \eps}$ by greedily inserting items from $H_{< \eps} = \{r_i \in H \mid w_i < \eps\}$ into this stack. To do this, start with $H' = H_{\geq \eps}^*$. Now sort the items of $H_{< \eps}$ by non-increasing order of height and for each item try to insert it into the stack (at the correct position to preserve the order inside the stack). If this is possible, the item is added to $H'$. 

Assume that $H_{< \eps} = \{r_1, \ldots, r_m\}$ with $h_1 \leq \cdots \leq h_m$. 
Let $v_i = \width{w}(\{r_j \in H' \mid h_j \geq h_i\})$ and $v_i^* = \width{w}(\{r_j \in H^* \mid h_j \geq h_i\})$. Whenever an item $r_i$ is not inserted in the stack we have $v_i > v_i^* - w_i$. To see this, assume that $v_i^* \geq v_i + w_i$. This means that the substack from $H^*$ of items of height at least $h_i$ has width larger than the substack of items of height at least $h_i$ from the stack of $H'$ \emph{plus} the width of $r_i$. Thus $r_i$ does not cause a conflict. Now it is easy to see by induction that $\width{w}(H') > \width{w}(H^*) - \eps$ at the end.
\qed
\end{proof}

With these preparations, the following lemma is easy to show.

\begin{lemma} \label{lem:largeW}
Let $\eps > 0$ and let $I$ be an instance with $\Opt{I}Ê= 1$, $\height{h}(W) \geq \width{w}(H)  > 1/2$, and $\height{h}(W_{\delta}) > (\delta  - \eps)/(1+2\delta)$ and $\width{w}(H_{\delta}) > (\delta  - \eps)/(1+2\delta)$ for all $\delta \in (\eps, 1/2]$. There exists a polynomial-time algorithm that returns a packing of $I$ into two bins.
\end{lemma}

\begin{proof}
Use Lemma~\ref{lem:PackWH} to pack $W \cup H'$ with $H' \subseteq H$ and $\width{w}(H') > \width{w}(H)/2 - \eps$ in the first bin. Build a stack of the remaining high items $H\setminus H'$ and align it with the left side of the second bin. The width of this stack is $\width{w}(H\setminus H') < \width{w}(H)/2 + \eps$. Note that $\width{w}(H\setminus H') \leq 1/2$, as otherwise $\height{h}(W) \geq \width{w}(H) \geq 1-2\eps$, which gives a contradiction since by Lemma~\ref{lem:areaofWH} we would then have $\Vol{W \cup H} \geq 2\xi + (\width{w}(H) + \height{h}(W))/2 \geq 2\xi + 1-2\eps > 1$.
Pack the remaining items $T$ with Steinberg's algorithm in the free rectangle of size $(a,b)$ with $a = 1-\width{w}(H\setminus H')$ and $b=1$ to the right of the second bin. This is possible since $w_{\max}(T) \leq 1/2 \leq 1-\width{w}(H\setminus H')$, $h_{\max}(T) \leq 1/2$ and with Lemma~\ref{lem:areaofWH} we have 
\begin{align*}
2\Vol{T} & \leq 2\Big(1-2\xi -\frac{\width{w}(H) + \height{h}(W)}{2}\Big) \\
& \leq 2 - 4\xi - \frac{\width{w}(H)}{2} - \frac{\width{w}(H)}{2} - \frac{1}{2} && \textrm{as $\height{h}(W)Ê\geq 1/2$} \\
& < \frac{3}{2}- 4\xi - \width{w}(H\setminus H')+ \eps - \frac{1}{4} && \textrm{as $\frac{\width{w}(H)}{2} \geq\width{w}(H\setminus H') -\eps$ and $\frac{\width{w}(H)}{2} > \frac{1}{4}$} \\
& < 1-\width{w}(H\setminus H')&& \textrm{as $4\xi > \frac{1}{4} + \eps$} \\
& = ab - (2w_{\max}-a)_+(2h_{\max}-b)_+ \ && \textrm{as $2h_{\max} - b \leq 0$.}
\end{align*}
\qed
\end{proof}

In the following we assume that $\width{w}(H)Ê\leq 1/2$ as otherwise we could pack the instance into two bins with the algorithms of Lemma~\ref{lem:smallheight} or Lemma~\ref{lem:largeW}. Furthermore, we still have our initial assumption $\height{h}(W)Ê\geq \width{w}(H)$. The following lemma shows that certain sets of small items can be packed together with the set of wide items. 

\begin{lemma}\label{lem:partition}
Any set $T = \{r_1, \ldots, r_m\}$ where $r_i = (w_i, h_i)$ with $w_i \leq 1/2$, $h_i \leq 1-\height{h}(W)$ for $i=1, \ldots, m$ and total area $\Vol{T} \leq 1/2 -\height{h}(W)/2$ can be packed together with $W$.
\end{lemma}

\begin{proof}
Pack $W$ into a stack of height $\height{h}(W)$ and align this stack with the bottom of the bin. Use Steinberg's algorithm to pack $T$ into the free rectangle of size $(a,b)$ with $a= 1$ and $b  = 1-\height{h}(W)$ above $W$. This is possible since $w_{\max}(T) \leq 1/2$, $h_{\max}(T) \leq 1-\height{h}(W)$ and $2\Vol{T} \leq 1 - \height{h}(W) = ab = ab - (2w_{\max}-a)_+(2h_{\max}-b)_+$ as $2w_{\max}- a \leq 0$. \qed
\end{proof}

Obviously, Lemma~\ref{lem:partition} can also be formulated such that we pack the high items together with a set of small items of total area at most $1/2 - \width{w}(H)/2$ (in this case we do not need a condition like $h_i \leq 1-\height{h}(W)$, as $\width{w}(H) \leq 1/2$ and thus all remaining items fit into the free rectangle next to the stack of $H$). This suggests partitioning the small items into sets with these area bounds in order to pack them with the wide and high items. Before we can do this, we consider two cases where we have to apply a different packing. In all cases, we pack $W$ in a stack of height $\height{h}(W)$ in the bottom right corner of the first bin and pack $H$ in a stack of width $\width{w}(H)$ in the top left corner of the second bin. 

Let $\omega$ be the greatest width in the stack of $W$ that is packed above height $1/2$ in this packing (let $\omega=1/2$ if $\height{h}(W) \leq 1/2$)--see Figure~\ref{fig:DefineOmega}. 
We consider the set $\widetilde{H} = \{r_i \mid h_i \in (1-\height{h}(W), 1/2] \}$, i.e., the set of remaining items that do not fit above the stack of the wide items (and thus violate the condition of Lemma~\ref{lem:partition}). Since $\widetilde{H} = \emptyset$ for $\height{h}(W)Ê\leq 1/2$, the following case can only occur if $\height{h}(W) > 1/2$. 

\begin{figure}  [!tbp]
  \centering
  \includegraphics[width=.4\textwidth]{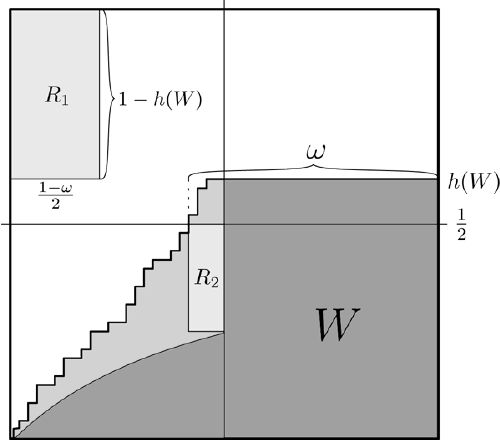}
  \caption{The greatest width in the stack of $W$ that is packed above height $1/2$ is denoted by $\omega$. The additional area guarantees for Case~1 of $\Vol{R_1} = (1-\height{h}(W))(1-\omega)/2$ and $\Vol{R_2} = (\omega-1/2)/4$ are drawn with lighter shading. The original area guarantee from Lemma~\ref{lem:areaofWH} is drawn with darker shading. \label{fig:DefineOmega}}
\end{figure}

\textbf{Case 1.} $\width{w}(\widetilde{H}) \geq (1-\omega)/2$.\\
If there is an item $r_i = (w_i,h_i) \in \widetilde{H}$ with $w_i > (1-\omega)/2$ then we pack this item in the top left corner of the first bin. It is easy to see that this is possible since $\Opt{I} = 1$. Otherwise greedily pack items from $\widetilde{H}$ into a horizontal stack in the top left corner of the first bin as long as they fit. Since all items in $\widetilde{H}$ have width at most $(1-\omega)/2$ we can pack a total width of at least $(1-\omega)/2$ before a conflict occurs. 

In both cases we packed a total area of at least $(1-\height{h}(W))(1-\omega)/2$ from $\widetilde{H}$ into the first bin--see rectangle $R_1$ in Figure~\ref{fig:DefineOmega}. Furthermore, we can use the definition of $\omega$ to improve the estimate from Lemma~\ref{lem:areaofWH}. In Lemma~\ref{lem:areaofWH} all wide items that are packed above height $1/4$ only contribute with their trivial area guarantee of half their height. Now we know that the items that are packed between height $1/4$ and $1/2$ have width at least $\omega$ and this gives an additional area of at least $(\omega-1/2) \cdot 1/4$--see rectangle $R_2$ in Figure~\ref{fig:DefineOmega}.
Thus we packed an overall area of 
\begin{align*}
A & \geq 2\xi + \frac{\width{w}(H) + \height{h}(W)}{2}Ê+ \frac{\omega-1/2}{4} + \frac{(1-\height{h}(W))(1-\omega)}{2} \\
& = \frac{3}{8} + 2\xi + \frac{\width{w}(H)}{2} + \omega \Big(\frac{\height{h}(W)}{2}Ê- \frac{1}{4}\Big) \\
& \geq \frac{3}{8} + 2\xi + \frac{\width{w}(H)}{2} && \textrm{as $\height{h}(W) \geq 1/2$ and $\omega > 0$.} 
\end{align*}
Therefore, the remaining items have total area at most $5/8 - 2\xi - \width{w}(H)/2 < 1/2 - \width{w}(H)/2$. Thus Lemma~\ref{lem:partition} allows us to pack these items together with the high items in the second bin.

Let $r', r''$ be the two largest items in $S\setminus \widetilde{H}$, i.e., among the remaining small items $S$ with $h_i \leq 1- \height{h}(W)$.

\textbf{Case 2.} $\Vol{\{r',r''\}} \geq 1/2 - 2\xi - \height{h}(W)/2$. \\
Pack $r'$ in the top left and $r''$ in the top right corner of the first bin. This is possible as the width of both items is at most $1/2$ and the height is at most $1-\height{h}(W)$ for $r',r''$. By Lemma~\ref{lem:areaofWH} and the condition for this case the total area of the packed items is
\begin{align*}
A & \geq 2\xi + \frac{\width{w}(H) + \height{h}(W)}{2} + \frac{1}{2} - 2\xi - \frac{\height{h}(W)}{2}  \geq  \frac{1}{2} + \frac{\width{w}(H)}{2}.
\end{align*}
Again we can use the method of Lemma~\ref{lem:partition} to pack the remaining items together with $H$ in the second bin.

\textbf{Case 3.} Otherwise we have $\width{w}(\widetilde{H}) < (1-\omega)/2$ and $\Vol{\{r',r''\}} < 1/2 - 2\xi - \height{h}(W)/2$. This yields that $\Vol{\widetilde{H}} < (1-\omega)/4$ and $\Vol{\{r''\}} < 1/4 - \xi - \height{h}(W)/4$ (where we assume that $\Vol{\{r'\}} \geq \Vol{\{r''\}}$. Use the following greedy algorithm to partition the remaining items into two sets that will be packed together with $W$ and $H$ using Lemma~\ref{lem:partition}. 
\begin{enumerate}
\item Create sets $S_1$ and $S_2$ with capacities $c_1 = 1/2 - \height{h}(W)/2$ and $c_2 = 1/2 - \width{w}(H)/2$, respectively,
\item add $r'$ to $S_1$ and add all items of $\widetilde{H}$ to $S_2$, \label{step:largeitems}
\item take the remaining items by non-increasing order of size and greedily add them to the set of greater remaining free capacity, i.e., to a set with maximal $c_i - \Vol{S_i}$.\label{step:greedycapacity}
\end{enumerate}
In the following we show that $\Vol{S_1} \leq c_1$ and $\Vol{S_2}Ê\leq c_2$. First note that this holds after Step~\ref{step:largeitems} since 
\begin{align*}
\Vol{\{r'\}} &< \frac{1}{2} - 2\xi - \frac{\height{h}(W)}{2} < \frac{1}{2} - \frac{\height{h}(W)}{2}Ê= c_1 \quad \textrm{ and } \\
\Vol{\widetilde{H} } &< \frac{1-\omega}{4} < \frac{1-\omega}{2} \leq \frac{1}{2} - \frac{\width{w}(H)}{2} = c_2 && \textrm{as $\omega \geq 1/2 \geq \width{w}(H)$.}
\end{align*}

Assume that after Step~\ref{step:greedycapacity} one of the sets has total area greater than its corresponding capacity. Then there is an item $r^*$ that has been added in Step~\ref{step:greedycapacity} and that violates the capacity for the first time. Since $\height{h}(W) \geq 1/4-\eps/2$ by Inequality~(\ref{eqn:LowerBoundTotalH}) and $\eps \leq 1/4$ we have $\Vol{\{r^*\}} \leq \Vol{\{r''\}} \leq 1/4 - \xi - \height{h}(W)/4 \leq 3/16 - \xi + \eps/8 < 0.15$. Assume w.l.o.g.\ that $r^*$ was added to $S_1$. Then we have $\Vol{S_1} > c_1$ and $\Vol{S_2} \geq c_2 - \Vol{r^*}$ as otherwise $r^*$ would have been added to $S_2$. Thus we have
\begin{align*}
\Vol{S_1} + \Vol{S_2} > c_1 + c_2 - \Vol{\{r^*\}} \geq c_1 + c_2 - 0.15.
\end{align*}
With Lemma~\ref{lem:areaofWH} we get the contradiction
\begin{align*}
\Vol{S_1} + \Vol{S_2}  & \leq 1- \Vol{W \cup H} \\
& \leq 1 - 2\xi - \frac{\width{w}(H) + \height{h}(W)}{2} \\
&= c_1 + c_2 - 2\xi \\
&= c_1 + c_2 - 0.15.
\end{align*}
Thus both sets do not violate their capacities and we can use the methods of Lemma~\ref{lem:partition} to pack $S_1$ together with $W$ in the first bin and $S_2$ together with $H$ in the second bin. We showed the following lemma.
\begin{lemma}\label{lem:smallW}
Let $\eps > 0$ and let $I$ be an instance with $\Opt{I}Ê= 1$, $\width{w}(H) \leq 1/2$, and $\height{h}(W_{\delta}) > (\delta  - \eps)/(1+2\delta)$ and $\width{w}(H_{\delta}) \leq (\delta  - \eps)/(1+2\delta)$ for all $\delta \in (\eps, 1/2]$. There exists a polynomial-time algorithm that returns a packing of $I$ into two bins.
\end{lemma}

This concludes our algorithm for instances $I$ with $\Opt{I} = 1$ as one of the Lemmas~\ref{lem:smallheight}, \ref{lem:largeW} and \ref{lem:smallW} can pack $I$ into two bins. In the next section we show how to handle instances with $\Opt{I} > 1$.

\section{Packing instances that fit into a constant number of bins}
\label{sec:OptConst}
In the following we describe our algorithm that packs the instances $I$ with $2 \leq \Opt{I} < k$ into $2 \Opt{I}$ bins. Let $\eps = 1/(40k^3 +2)$.  

Let $L = \{r_i \mid w_i h_i > \eps\}$ be the set of \emph{large} items and let $T = \{r_i \mid w_i h_i \leq \eps\}$ be the set of \emph{tiny} items. As defined in Section~\ref{sec:tools} we refer to items as wide ($W$), high ($H$), small ($S$) and big, according to their side lengths. Note that the terms \emph{large} and \emph{tiny} refer to the area of the items whereas the predicates \emph{big}, \emph{wide}, \emph{high} and \emph{small} are given depending on the width and height of the items. Also note that, e.g., an item can be tiny and high, or wide and big at the same time.

We guess $\ell = \Opt{I}$ and open $2\ell$ bins that we denote by $B_1, \ldots, B_\ell$ and $C_1, \ldots, C_\ell$. By \emph{guessing} we mean that we iterate over all possible values for $\ell$ and apply the remainder of this algorithm on every value. As there are only a constant number of values, this is possible in polynomial time. We assume that we know the correct value of $\ell$ as we eventually consider this value in an iteration. For the ease of presentation, we also denote the sets of items that are associated with the bins by $B_1, \ldots, B_\ell$ and $C_1, \ldots, C_\ell$. We will ensure that the set of items that is associated with a bin is feasible and a packing is known or can be computed in polynomial time. 

Let $I_i^*$ be the set of items in the $i$-th bin in an optimal solution. We assume w.l.o.g. that $\Vol{I_i^*} \geq \Vol{I_j^*}$ for $i < j$. Then we have
\begin{align}
\Vol{I} = \Vol{I_1^*} + \cdots + \Vol{I_\ell^*} & \leq \ell \cdot \Vol{I_1^*} \label{eqn:volI}
\end{align}

In a first step, we guess the assignment of the large items to bins. Using this assignment and the BCS algorithm we pack a total area of at least $\Vol{I_1^*} - \eps$ into $B_1$ and keep $C_1$ empty. This step has the purpose of providing a good area bound for the first bin (especially since it is set into relation to the fullest bin of an optimal solution in Inequality~(\ref{eqn:volI})) and leaving a free bin for later use. We use a slightly technical definition (which is given later in the details) for the profits of the items to ensure that the large items that are assigned to $B_1$ are actually packed. For all other bins we reserve $B_i$ for the wide and small items (except the big items) and $C_i$ for the high and big items for $i = 2, \ldots, \ell$. This separation enables us to use Steinberg's algorithm to pack up to half of the bins' area. In detail, the first part of the algorithm works as follows.

\begin{enumerate}
\item Guess $L_i = I_i^* \cap L$ for $i = 1, \ldots, \ell$.
\item Apply the BCS algorithm on $L_1 \cup T$ with $p_i = \Vol{r_i} (1/\eps + 1)$ for $r_i \in L_1$ and $p_i = \Vol{r_i}$ for $r_i \in T$ using an accuracy of $\eps^2/(1+2\eps)$. Assign the output to bin $B_1$ and keep an empty bin $C_1$.
\item For $i = 2, \ldots, \ell$, assign the wide and small items of $L_i$ to $B_i$ (omitting big items) and assign the high and big items of $L_i$ to $C_i$. That is $B_i = L_i \setminus H$ and $C_i = L_i \cap H$. \label{step:separation}
\item For $i=2, \ldots, \ell$, greedily add tiny wide items from $T \cap W$ by non-increasing order of width to $B_i$ as long as $\Vol{B_i} \leq 1/2$ and greedily add tiny high items from $T \cap H$ by non-increasing order of height to $C_i$ as long as $\width{w}(C_i) \leq 1$. \label{step:greedy}
\end{enumerate}

Corollary~\ref{cor:big} shows that using Steinberg's algorithm the bins $B_2, \ldots, B_\ell$ can be packed as there are no wide items and the total area is at most $1/2$. The bins $C_2, \ldots, C_\ell$ can be packed with a simple stack as they contain only high items of total width at most 1.
Observe that in Step~\ref{step:greedy} we only add to a new bin $B_i$ if the previous bins contain items of total area at least $1/2-\eps$ and we only add to a new bin $C_i$ if the previous bins contain items of total width at least $1-2\eps$ (as the width of the tiny high items is at most $2\eps$) and thus of total area at least $1/2 (1-2\eps) = 1/2-\eps$. After the application of this first part of the algorithm, some tiny items $T' \subseteq T$ might remain unpacked. Note that if $\Vol{B_\ell} < 1/2 - \eps$, then there are no wide items in $T'$ and if $\Vol{C_\ell} < 1/2 -\eps$ then there are no high items in $T'$ (as these items would have been packed in Step~\ref{step:greedy}). We distinguish different cases to continue the packing according to the filling of the last bins $B_\ell$ and $C_\ell$. Prior to this we prove the area bound for $B_1$ that we already mentioned above.

First note that Theorem~\ref{thm:generalPTAS} can be applied as $p_i/\Vol{r_i} \in \{1,1/\eps + 1\}$ for all items in $L_1 \cup T$. Recall that we used the technical definition of the profits to ensure that $L_1$ is completely packed, i.e., $B_1 \cap L = I_1^* \cap L$. This is easy to see since $p_i > 1+\eps$ for $r_i \in L_1$, whereas $p(\widetilde{T}) = \Vol{\widetilde{T}} \leq 1$ for any feasible set $\widetilde{T} \subseteq{T}$. Furthermore, for the set of packed tiny items $B_1 \cap T$ we have 
\begin{align*}
\Vol{B_1 \cap T} \geq \Vol{I_1^* \cap T}Ê- \eps
\end{align*}
since $(1/\eps + 1)\Vol{B_1 \cap L} + \Vol{B_1 \cap T} = p(B_1)$ and
\begin{align*}
p(B_1) & \geq \Big(1-\frac{\eps^2}{1+2\eps}\Big)\Opt{L_1 \cup T} - \frac{\eps^2}{1+2\eps} && \textrm{by Theorem~\ref{thm:generalPTAS}} \\
&\geq \Big(1-\frac{\eps^2}{1+2\eps}\Big)\Big[\Big(\frac{1}{\eps} + 1\Big)\Vol{I_1^* \cap L} + \Vol{I_1^* \cap T}\Big] - \frac{\eps^2}{1+2\eps} \\
& = \Big(\frac{1}{\eps} + 1\Big)\Vol{I_1^* \cap L} + \Vol{I_1^* \cap T} - \Big(\frac{1}{\eps}+2\Big) \frac{\eps^2}{1+2\eps} \\
&= \Big(\frac{1}{\eps} + 1\Big)\Vol{B_1 \cap L} + \Vol{I_1^* \cap T} - \eps.
\end{align*}
Thus we have
\begin{align}
\Vol{B_1} \geq \Vol{I_1^*}Ê- \eps. \label{eqn:B1}
\end{align}
Now we are ready to start with the case analysis.

\textbf{Case 1.} $\Vol{B_\ell} < 1/2 -\eps$ and $\Vol{C_\ell} < 1/2 - \eps$. \\
In this case $T'$ does not contain any wide or high items as these items would have been packed to $B_\ell$ or $C_\ell$. Greedily add items from $T'$ into all bins except $B_1$ as long as the bins contain items of total area at most $1/2$. After adding the items from $T'$, either all items are assigned to a bin (and can thus be packed) or each bin contains items of total area at least $1/2 - \eps$ and we packed a total area of at least
\begin{align*}
A & \geq \Vol{B_1} + (2\ell - 1) \Big(\frac{1}{2} - \eps\Big) \\
& \geq \Vol{I_1^*} + \ell - \frac{1}{2}- 2\ell\eps && \textrm{from Inequality~(\ref{eqn:B1})} \\
& \geq \Vol{I_1^*}  + \ell\Vol{I_1^*} + \ell(1-\Vol{I_1^*}) - \frac{1}{2}- 2\ell\eps \\
& \geq \ell\Vol{I_1^*} + \frac{1}{2} - 2\ell\eps && \textrm{as $\ell \geq 1$ and $1-\Vol{I_1^*} \geq 0$} \\
& > \ell\Vol{I_1^*} && \textrm{as $\eps < \frac{1}{4\ell}$}.
\end{align*}
Since this contradicts $AÊ\leq \Vol{I} \leq \ell \Vol{I_1^*}$ from Inequality~(\ref{eqn:volI}), all items are packed.

\textbf{Case 2.} $\Vol{B_\ell} \geq 1/2 -\eps$ and $\Vol{C_\ell} \geq 1/2 - \eps$. \\
In this case $T'$ might contain wide and high items. On the other hand the bin $C_1$ is still available for packing. We use the area of the items in bin $C_\ell$ to bound the total area of the packed items. With a calculation as in Case 1 we get
\begin{align*}
A & \geq  \Vol{B_1} + \Vol{C_\ell} + \overbrace{(2\ell - 3) \Big(\frac{1}{2} - \eps\Big)}^{\textrm{all bins except $B_1, C_1,C_\ell$}}\\
& \geq \ell\Vol{I_1^*} + \Vol{C_\ell} - \frac{1}{2} - (2\ell-2)\eps.
\end{align*}
Again from Inequality~(\ref{eqn:volI}) we have $\Vol{I} \leq \ell\Vol{I_1^*}$. Thus we get 
\begin{align}
\Vol{T'} &Ê\leq \Vol{I}-A \leq \frac{1}{2} + (2\ell -2)\eps - \Vol{C_\ell}  \quad \textrm{and hence} \label{eqn:Cell1} \\
\Vol{T'} &Ê\leq (2\ell - 1)\eps && \textrm{as $\Vol{C_\ell} \geq 1/2-\eps$} \label{eqn:Cell2}
\end{align} 
By Inequality~(\ref{eqn:Cell1}) we know that if $\Vol{C_\ell} > 1/2 + (2\ell-2)\eps$ then $\Vol{T'} < 0$ and thus all items are packed. Therefore we assume that $\Vol{C_\ell} \leq 1/2 + (2\ell-2)\eps$.

We consider the set $\widehat{H} = \{r_i \in C_\ell \mid h_i \leq 3/4\}$. If $\width{w}(\widehat{H}) \geq (4\ell-3)\eps$ then remove $\widehat{H}$ from $C_\ell$ and pack it in a stack in $C_1$ instead. As we now have $\Vol{C_\ell} \leq 1/2 - (2\ell-1)\eps$ and $\Vol{T' \setminus W} \leq \Vol{T'} \leq (2\ell-1)\eps$ by Inequality~(\ref{eqn:Cell2}), we can pack $T' \setminus W$ together with $C_\ell$. The remaining items $T' \cap W$ have total height at most $2(2\ell-1)\eps$ and thus fit above $\widehat{H}$ into $C_1$.

Otherwise, there is no item $r' = (w',h')$ in $C_\ell$ with $h' \leq 3/4$ and $w' \geq (4\ell-3)\eps$. Let $\widetilde{H} = \{r_i \in C_\ell \cup T' \mid h_i > 3/4\}$. Observe that we have
\begin{align*}
\width{w}(\widetilde{H}) \leq \frac{\Vol{C_\ell \cup T'}}{3/4} & \leq \frac{4}{3}Ê\Big(\frac{1}{2} + (2\ell-2)\eps + (2\ell-1)\eps\Big)  = \frac{2}{3} + \Big(\frac{16}{3}\ell -4\Big)\eps < 1.
\end{align*}
We take \emph{all} high items from $C_\ell \cup T'$ and order them by non-increasing height. Now pack the items greedily into a stack of width up to $1$ and pack this stack into $C_\ell$. We have $\width{w}(C_\ell) \geq 1-(4\ell-3)\eps$ as the total width of the items from $\widetilde{H}$ is bounded and thus all further items have width at most $(4\ell-3)\eps$ (as otherwise $\widehat{H} \geq (4\ell-3)\eps$ and we had solve the problem in the previous step). For the remaining items $T'$ we have $h_{\max}(T') \leq 3/4$ and $\Vol{T'} \leq 1/2 - (2\ell-2)\eps - \Vol{C_\ell} \leq (4\ell-7/2) \eps \leq 1/4$ (by Inequality~\ref{eqn:Cell1} and as $\Vol{C_\ell} \geq \width{w}(C_\ell)/2 \geq 1/2 - (2\ell-3/2)\eps$ and $\eps \leq 1/(16\ell)$. Thus $T'$ can be packed into bin $C_1$ using Steinberg's algorithm.

\textbf{Case 3.} $\Vol{B_\ell} < 1/2 -\eps$ and $\Vol{C_\ell} \geq 1/2 - \eps$. \\
If $\width{w}(T' \cap H) \leq 1$ then pack $T' \cap H$ in $C_1$ and proceed as in Case 1. 

The subcase where $\width{w}(T' \cap H) > 1$ is the most difficult of all four cases. The challenge that we face is that $\width{w}(H)$ can be close to $\ell$ (which is a natural upper bound) but we can only ensure a packed total width of at least $\ell(1 - 2\eps)$ in the bins $C_1, \ldots, C_\ell$. So we have to pack high items into the bins $B_2, \ldots, B_\ell$. We distinguish two further subcases.
\begin{enumerate}
\item Assume that there exists $j \in \{2, \ldots, \ell\}$ with $\width{w}(L_j \cap H) > 10\ell\eps$, i.e., the total width of the items that are large and high and associated with the $i$-th bin in an optimal packing is large enough such that moving these items away gives sufficient space for the still unpacked high items.

Go back to Step~\ref{step:separation} in the first part of the algorithm and omit separating the items from $L_j$. Instead we assign the items from $L_j$ to bin $B_j$ and keep $C_j$ free at the moment. Note that $L_j$ admits a packing into a bin as $L_j$ corresponds to the large items in a bin of an optimal solution. Since $|L_j| \leq 1/\eps$ we can find such a packing in constant time.

While greedily adding tiny items in Step~\ref{step:greedy}, we skip $B_j$ for the wide items and we continue packing high items in $C_1$ after we have filled $C_2, \ldots, C_\ell$. As we moved high items of total width at least $10\ell\eps$ to $B_j$ and we can pack high items of total width at least $1-2\eps$ into each bin, no high items remains after this step. Finally, greedily add remaining tiny items to bins $B_2, \ldots, B_\ell$ except $B_j$, using the area bound $1/2$.

Now consider the bins $C_1$ and $C_j$. Both contain only tiny items, as we moved the large items from $C_j$ to $B_j$. We packed the tiny items greedily by height and thus all items in $C_j$ have height greater or equal to any item in $C_1$. Let $h'$ be greatest height in $C_1$. Then we have $\Vol{C_j}Ê\geq h' (1-2\eps)$. Furthermore, we know that $\width{w}(H) > \ell(1-2\eps)$. Thus we have 
\begin{align*}
\Vol{H}Ê> (\ell-1)(1/2-\eps) + h' (1-2\eps).
\end{align*}
If after the modified Step~\ref{step:greedy} tiny items remain unpacked, then all bins $B_i$ for $i \in \{2, \ldots, \ell\}\setminus \{j\}$ have area $\Vol{B_i} \geq 1/2 - \eps$. By summing up the area of the high items separately we get a total packed area of at least
\begin{align*}
A & \geq \Vol{B_1} + \overbrace{(\ell-2)\Big(\frac{1}{2}-\eps\Big)}^{B_i \textrm{ for } i \in \{2, \ldots, \ell\}\setminus \{j\}} + \Vol{H} \\
& > \Vol{B_1} + (\ell-2)\Big(\frac{1}{2}-\eps\Big) + (\ell-1)\Big(\frac{1}{2}-\eps\Big) + h' (1-2\eps)\\
& \geq \Vol{I_1^*} + \ell -\frac{3}{2} + h' - (2\ell - 2 + 2h')\eps && \textrm{by Inequality~(\ref{eqn:B1})} \\
& \geq \Vol{I_1^*} + \ell\Vol{I_1^*} + \ell(1-\Vol{I_1^*}) - \frac{3}{2} + h' - 2\ell\eps && \textrm{as $h' \leq 1$}\\
& \geq \ell\Vol{I_1^*} - \frac{1}{2} + h' - 2\ell\eps && \textrm{as $\ell \geq 1$ and $1-\Vol{I_1^*} \geq 0$}. 
\end{align*}
On the other hand we have $A \leq \Vol{I}  \leq \ell\Vol{I_1^*}$ by Inequality~(\ref{eqn:volI}). Thus the total area of the remaining items $T'$ is at most $\Vol{T'}Ê\leq 1/2 + 2\ell\eps - h'$. If $h' \geq 1/2 + 2\ell\eps$ we packed all items. 

Otherwise we have $1/2 < h' < 1/2+2\ell\eps$ and 
\begin{align}
\Vol{T'} \leq 2\ell\eps.\label{eqn:arearemainingitems}
\end{align}
We will pack $T'$ in $C_1$ together with the already packed high items. Observe that 
\begin{align}
\width{w}(C_1 \cap H) \leq 1-8\ell\eps \label{eqn:widthC_1H}
\end{align}
as we move high items of total area of at least $10\ell\eps$ to $B_j$ and all bins $C_2, \ldots, C_\ell$ are filled up to a width of at least $1-2\eps$.

We pack the remaining items $T'$ into three rectangles $R_1 = (1, 4\ell\eps)$, $R_2 = (8\ell\eps, 1-4\ell\eps)$ and $R_3 = (1 - 8\ell\eps, 1/2-6\ell\eps)$ which can be packed in $C_1$ together with $C_1 \cap H$ as follows--see Figure~\ref{fig:Case2}. Pack the stack of $C_1 \cap H$ in the lower left corner and pack $R_3$ above this stack. As $h' + \height{h}(R_3) \leq 1 - \height{h}(R_1)$, $R_1$ fits in the top of $C_1$. Finally, pack $R_2$ in the bottom right corner. This is possible as $\height{h}(R_2) \leq 1-\height{h}(R_1)$ and $\width{w}(C_1 \cap H) \leq 1-8\ell\eps = 1-\width{w}(R_3)$ by Inequality~(\ref{eqn:widthC_1H}).

\begin{figure}  [!tbp]
  \centering
  \includegraphics[width=.33\textwidth]{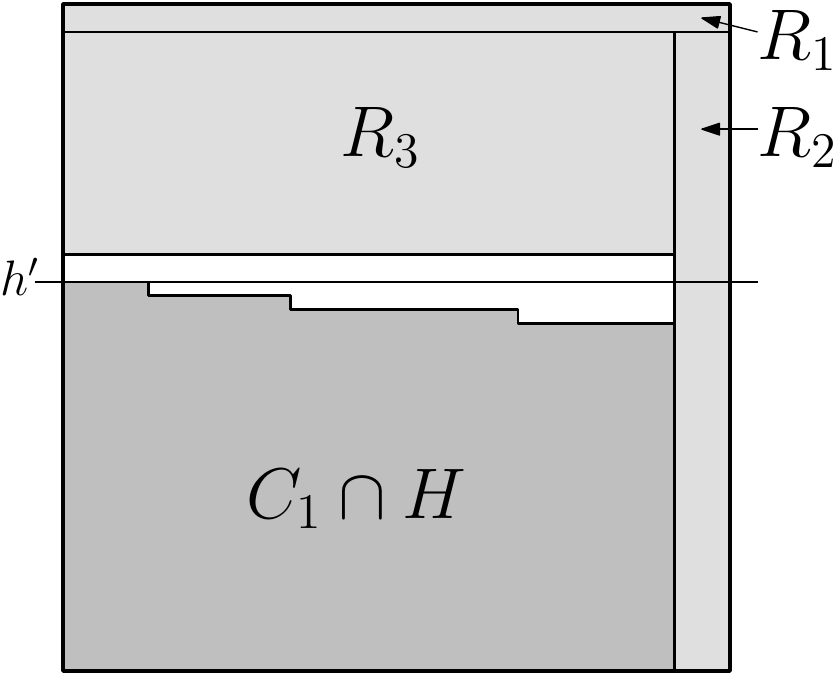}
  \caption{The three rectangles $R_1$, $R_2$ and $R_3$ for packing $T'$ together with $C_1 \cap H$ in $C_1$. \label{fig:Case2}}
\end{figure}

Now pack $T' \cap W$ in a stack in $R_1$ (which is possible since $\height{h}(T'\cap W) \leq 2\Vol{T'} \leq 4\ell\eps$ by Inequality~(\ref{eqn:arearemainingitems})) and pack all $r_i = (w_i, h_i) \in T'$ with $h_i > 1/2-6\ell\eps$ in a vertical stack in $R_2$ (this fits as the total width of items with $h_i > 1/2-6\ell\eps$ is at most $\Vol{T'} / (1/2-6\ell\eps) \leq 8\ell\eps$ by Inequality~(\ref{eqn:arearemainingitems}) and as $6\ell\eps \leq 1/4$). Finally, use Steinberg's algorithm to pack the remaining items in $R_3$. This is possible since $w_{\max} \leq 1/2$, $h_{\max} \leq 1/2 - 6 \ell\eps$ and 
\begin{align*}
2\Vol{T'} \leq 4\ell\eps & \leq  \frac{1}{2}-14\ell\eps + 96 \ell^2\eps^2\\
& = (1-8\ell\eps)\Big(\frac{1}{2}-6\ell\eps\Big) - (1-1+8\ell\eps)_+\Big(1-12\ell\eps - \frac{1}{2}+6\ell\eps\Big)_+.
\end{align*}

This finishes the first case where we assumed that there exists $j \in \{2, \ldots, \ell\}$ with $\width{w}(L_j \cap H) > 10\ell\eps$.

\item Now assume that we have $\width{w}(L_j \cap H) \leq 10\ell\eps$ for all $j \in \{2, \ldots, \ell\}$ and in particular $w_i \leq 10\ell\eps$ for all items $r_i = (w_i, h_i) \in (L_2 \cup \cdots \cup L_\ell)$.
Thus all high items that are not packed in $B_1$ are thin, i.e., have width at most $10\ell\eps$. We use this fact by repacking the high items greedily by non-increasing height in the bins $C_1, \ldots, C_\ell$. Each bin contains high items of total width at least $1-10\ell\eps$ afterwards. Thus high items of total width at most $10\ell^2\eps$ remain unpacked. This is worse than in the previous case but since we repacked all items we can get a nice bound on the height of the unpacked items. Let $h'$ be the smallest height in $C_\ell$. Then all items in $C_1, \ldots, C_\ell$ have height at least $h'$ and the remaining items $T'$ have height at most $h'$. If there is an $i \in \{2, \ldots, \ell\}$ with $\Vol{B_i} \leq 1- h' - 10\ell^2\eps$ then we can add the remaining items $T'$ to $B_i$ using Steinberg's algorithm. To see this note that $h_{\max}(T' \cup B_i) \leq h'$ and $2\Vol{T'  \cup B_i} \leq 2 \Vol{B_i} + 2(10\ell^2\eps)h' \leq 2-2h'$ which corresponds to the bound from Theorem~\ref{thm:Steinberg}.

Otherwise for all $i \in \{2, \ldots, \ell\}$ we have $\Vol{B_i} \geq 1-h'-0\ell^2\eps$. Then we packed a total area of at least
\begin{align*}
A & \geq \Vol{B_1} + \Vol{C_1 \cup \cdots \cup C_\ell} + \Vol{B_2 \cup \cdots \cup B_\ell} \\
& \geq \Vol{B_1} + h' \ell (1-10\ell\eps) + (\ell-1)(1-h'-10\ell^2\eps ) \\
& \geq \Vol{I_1^*} + \ell-1 + h' - (10\ell^2 + 10\ell^2 + 10\ell^3+1)\eps \\
& > \Vol{I_1^*} + \ell\Vol{I_1^*} + \ell(1-\Vol{I_1^*})-1 + \frac{1}{2} - (20\ell^3 +1)\eps && \textrm{as $\ell \geq 2$ and $h' > \frac{1}{2}$} \\
& \geq \ell\Vol{I_1^*}Ê+ \frac{1}{2} - (20\ell^3 +1) \eps && \textrm{as $\ell \geq 1$ and $1-\Vol{I_1^*} \geq 0$}.
\end{align*}
And since $1/2 - (20\ell^3 +1)\eps \geq 0$ and $A \leq \Vol{I} \leq \ell\Vol{I_1^*}$, no item remains unpacked.
\end{enumerate}
Thus in both subcases we were able to derive a feasible packing.

\textbf{Case 4.} $\Vol{B_\ell} \geq 1/2 -\eps$ and $\Vol{C_\ell} < 1/2 - \eps$.\\
In this case $T'$ contains no high items. If there are also no wide items remaining in $T'$, apply the methods of Case 1. Otherwise we use the following process to free some space in the bins for wide and small items, i.e., $B_2, \ldots, B_\ell$. The idea of the process is to move small items from bins $B_i$ to bins $C_i$ and thereby move the tiny high items $T' \cap H$ further in direction $C_\ell$. To do this, let $S_i = L_i \cup S$ be the set of small items in $B_i$. 

Remove the tiny items from $C_2, \ldots, C_\ell$. If there exists an item $r \in S_i \cap B_i$ for some $i \in \{2, \ldots, \ell\}$ then remove $r$ from $B_i$ and add it to $C_i$, otherwise stop. Adding $r$ to $C_i$ is possible as $C_i$ is a subset of $L_i = I_i^*$ and thus feasible. Add wide items from $W \cap T'$ to $B_i$ until $\Vol{B_i} \geq 1/2 - \eps$ again or $W \cap T' = \emptyset$. Finally, add the high items from $H \cap T'$ to $C_2, \ldots, C_\ell$ in a greedy manner analogously to Step~\ref{step:greedy} of the first part of the algorithm but using the area bound $\Vol{C_i} \leq 1/2$. This ensures that all sets $C_i$ can be packed with Steinberg's algorithm. Repeat this process until $S_i \cap B_i = \emptyset$ for all $i \in \{2, \ldots, \ell\}$ or $T'$ contains a high item at the end of an iteration.

There are two ways in which this process can stop. First if we moved all items from $S_i$ to $C_i$, and second if in the next step a high item would remain in $T'$ after the process. In the first case we have reached a situation as in Case 2 or Case 3, i.e., the roles of the wide and the high items are interchanged and $\Vol{B_\ell} \geq 1/2 - \eps$. Thus by rotating all items and the packing derived so far, we can solve this case analogously to Case 2 or Case 3, depending on $\Vol{C_\ell}$. 

In the second case, let $r^*$ be the item that stopped the process, i.e., if $r^*$ is moved from $B_i$ to $C_i$ for some $i \in \{2,\ldots, \ell\}$, at least one high item would remain in $T'$. Then, instead of moving $r^*$ to $C_i$ we move $r^*$ to $C_1$ and add items from $T'$ to $C_1$ and $C_i$ as long as $\Vol{C_1} \leq 1/2$ and $\Vol{C_i} \leq 1/2$. The resulting sets can be packed with Steinberg's algorithm as no item has width greater than $1/2$. If after this step still items remain unpacked then a calculation similar to Case 1 gives a total packed area of
\begin{align*}
A & \geq  \Vol{B_1}  + \overbrace{(2\ell - 2) \Big(\frac{1}{2} - \eps\Big)}^{\textrm{all bins except $B_1, B_i$}} + \overbrace{\frac{1}{2} - \eps - \Vol{r^*}}^{\textrm{bin $B_i$}} \\
& \geq \ell\Vol{I_1^*} + \frac{1}{2} - 2\ell\eps - \Vol{r^*} > \ell\Vol{I_1^*} && \textrm{since $\Vol{r^*} \leq 1/4$ and $\eps < 1/(8\ell)$}.
\end{align*}
From Inequality~(\ref{eqn:volI}) we have the contradiction $A \leq \Vol{I} \leq \ell\Vol{I_1^*}$. Thus all items are packed.

Note that we use crossreferences between the four cases but there are no circles in these references, i.e., Case 3 uses Case 1, and Case 4 uses Cases 1, 2 and 3.

We showed the following lemma.

\begin{lemma}\label{lem:OptConst}
There exists a polynomial-time algorithm that, given an instances $I$ with $1 < \Opt{I} < k$, returns a packing in $2\Opt{I}$ bins. 
\end{lemma}

\section{The overall algorithm}
\label{sec:final}

Let us recapitulate the different cases of our algorithm. We use the asymptotic algorithm by Bansal et al.\citesuperscript{BansalCapraraSviridenko:2006a} that solves instances $I$ with $\Opt{I} \geq k$ for some constant $k$. As we do not know the optimal value in advance, we apply our algorithms in any case but each algorithm is allowed to fail if its requirement on $\Opt{I}$ is not satisfied. 

For $\Opt{I} = 1$ we presented an algorithm in Section~\ref{sec:Opt1} that returns a packing into two bins. The algorithm is based on two major cases according to the total height of the $\delta$-wide and the total width of the $\delta$-high items. Either these height (width) is suitably bounded (for some $\delta \in (\eps, 1/2]$) in which case we apply the methods of Lemma~\ref{lem:smallheight}, or we get a substantial area guarantee. We utilize this area guarantee in the methods of Lemmas~\ref{lem:largeW} and \ref{lem:smallW}.

For $1<\Opt{I} <k$ we presented an algorithm in Section~\ref{sec:OptConst} that is based on an enumeration of some large items and a separation of the wide and the high items. Lemma~\ref{lem:OptConst} shows that the algorithm outputs a packing into at most $2\Opt{I}$ bins. In total we showed the following theorem.

\begin{theorem}
There exists a polynomial-time 2-approximation algorithm for two-dimensional bin packing.
\end{theorem}

We cannot give an explicit running time of our algorithm as it is based on the BCS algorithm for which the running time is only stated as polynomial in $n$ for any fixed $\eps > 0$ and $r > 1$. The running time of our algorithm is also bounded by some polynomial in $n$.

Also note that in order to implement our algorithm one needs to know the constant $k$. It is not made explicit in in paper of Bansal et al.\citesuperscript{BansalCapraraSviridenko:2006a} but can in principle be bounded from above by their methods.

\bibliography{Archiv}

\begin{thebibliography}{10}

\bibitem{Bansal:2008}
N.~Bansal.
\newblock personal communication, 2008.

\bibitem{BansalCapraraSviridenko:2006a}
N.~Bansal, A.~Caprara, and M.~Sviridenko.
\newblock Improved approximation algorithms for multidimensional bin packing
  problems.
\newblock In {\em {FOCS}: {P}roc. 47nd {IEEE} {S}ymposium on {F}oundations of
  {C}omputer {S}cience}, pages 697--708, 2006.

\bibitem{BansalCapraraSviridenko:2008a}
N.~Bansal, A.~Caprara, and M.~Sviridenko.
\newblock A structural lemma in 2-dimensional packing, and its implications on
  approximability, 2008.
\newblock IBM Research Division, RC24468 (W0801-070),
  \url{http://domino.research.ibm.com/library/cyberdig.nsf/index.html}.

\bibitem{BansalCorreaKenyonSviridenko:2006a}
N.~Bansal, J.~R. Correa, C.~Kenyon, and M.~Sviridenko.
\newblock Bin packing in multiple dimensions - inapproximability results and
  approximation schemes.
\newblock {\em Mathematics of Operations Research}, 31(1):31--49, 2006.

\bibitem{Caprara:2002a}
A.~Caprara.
\newblock Packing 2-dimensional bins in harmony.
\newblock In {\em {FOCS}: {P}roc. 43rd {IEEE} {S}ymposium on {F}oundations of
  {C}omputer {S}cience}, pages 490--499, 2002.

\bibitem{CapraraLodiMonaci:2005a}
A.~Caprara, A.~Lodi, and M.~Monaci.
\newblock Fast approximation schemes for two-stage, two-dimensional bin
  packing.
\newblock {\em Mathematics of Operations Research}, 30(1):150--172, 2005.

\bibitem{ChlebikChlebikova:2005a}
M.~Chleb\'{\i}k and J.~Chleb\'{\i}kov\'a.
\newblock Inapproximability results for orthogonal rectangle packing problems
  with rotations.
\newblock In {\em {CIAC}: {P}roc. 6th {C}onference on {A}lgorithms and
  {C}omplexity}, pages 199--210, 2006.

\bibitem{HarrenvanStee:2009}
R.~Harren and R.~van Stee.
\newblock Absolute approximation ratios for packing rectangles into bins.
\newblock {\em Journal of Scheduling}, accepted for publication.

\bibitem{JansenThole:2008}
K.~Jansen and R.~Th{\"o}le.
\newblock Approximation algorithms for scheduling parallel jobs: Breaking the
  approximation ratio of 2.
\newblock In {\em {ICALP}: {P}roc. 35rd {I}nternational {C}olloquium on
  {A}utomata, {L}anguages and {P}rogramming}, pages 234--245, 2008.

\bibitem{JansenZhang:2007a}
K.~Jansen and G.~Zhang.
\newblock Maximizing the total profit of rectangles packed into a rectangle.
\newblock {\em Algorithmica}, 47(3):323--342, 2007.

\bibitem{LeungTamWongYoungChin:1990a}
J.~Y.-T. Leung, T.~W. Tam, C.~S. Wong, G.~H. Young, and F.~Y. Chin.
\newblock Packing squares into a square.
\newblock {\em Journal of Parallel and Distributed Computing}, 10(3):271--275,
  1990.

\bibitem{Schiermeyer:1994a}
I.~Schiermeyer.
\newblock Reverse-fit: A 2-optimal algorithm for packing rectangles.
\newblock In {\em {ESA}: {P}roc. 2nd {E}uropean {S}ymposium on {A}lgorithms},
  pages 290--299, 1994.

\bibitem{Steinberg:1997a}
A.~Steinberg.
\newblock A strip-packing algorithm with absolute performance bound 2.
\newblock {\em SIAM Journal on Computing}, 26(2):401--409, 1997.

\bibitem{Stee:2004a}
R.~van Stee.
\newblock An approximation algorithm for square packing.
\newblock {\em Operations Research Letters}, 32(6):535--539, 2004.

\bibitem{Zhang:2005a}
G.~Zhang.
\newblock A 3-approximation algorithm for two-dimensional bin packing.
\newblock {\em Operations Research Letters}, 33(2):121--126, 2005.

\end{thebibliography}
\bibliographystyle{abbrv}

\end{document}